\documentclass[10pt,conference, twocolumns]{IEEEtran}
\IEEEoverridecommandlockouts

\usepackage{cite}
\usepackage{amsmath,amssymb,amsfonts}
\usepackage{algorithmic}
\usepackage{graphicx}
\usepackage{textcomp}
\usepackage{xcolor}
\usepackage{color}
\usepackage{tikz}
\usepackage{verbatim}
\usepackage[style=base]{caption}
\usepackage{subcaption}

\usepackage{float} 
\usepackage{epsfig}
\usepackage{epstopdf}
\usepackage{wasysym}
\usepackage{dsfont}
\usepackage{amsthm}
\usepackage{mathtools} 
\usepackage{tabulary}
\allowdisplaybreaks% 
\usepackage{mathptmx}
\usepackage{bm}
\usepackage{subcaption}

\usepackage[english]{babel}
\usepackage{times} 
\usepackage{amsbsy}
\usepackage{balance} 
\usepackage{latexsym} 
\usepackage{txfonts}  
\usepackage{mathbbol} 
\usepackage{makecell}
\newcommand{\prob}{\mathrm{\textit{P}}}

\IEEEoverridecommandlockouts
\begin{document}

\setlength\abovedisplayskip{0pt}
\theoremstyle{plain}
\newtheorem{theorem}{Theorem} 
\newtheorem{lemma}{Lemma} 
\newtheorem{corollary}{Corollary}  
\newtheorem{remark}{Remark} 
\newcommand{\mynote}[2]{\textcolor{blue}{\fbox{\bfseries\sffamily\scriptsize#1}}
  \textcolor{blue}{{$/*$\textsf{\emph{#2}}$*/$}}}
%%%
\newcommand{\Yath}[1]{\mynote{Yathreb}{#1}}
\newcommand{\Fatm}[1]{\mynote{Fatma}{#1}}
\newcommand{\Juli}[1]{\mynote{Julie}{#1}}

\title{Spatiotemporal Modelling of Multi-Gateway LoRa Networks with Imperfect SF Orthogonality}

\author{\IEEEauthorblockN{Yathreb Bouazizi, Fatma Benkhelifa, Julie McCann\\
	\small  Imperial College London, London, UK\\
	\small \{y.bouazizi18,f.benkhelifa,j.mccann\}@imperial.ac.uk  \thanks{This work is supported by the Singapore Ministry of National Development and the National Research Foundation, Prime Minister’s Office under the Land and Liveability National Innovation Challenge (L2NIC) Research Programme (L2 NIC Award No. L2NICTDF1-2017-3). Any opinions, findings, and conclusions or recommendations expressed in this material are those of the author(s) and do not reflect the views of the Singapore Ministry of National Development and National Research Foundation, Prime Minister’s Office, Singapore.}}}

\maketitle
\begin{abstract}\
Meticulous modelling and performance analysis of Low-Power Wide-Area (LPWA) networks are essential for large scale dense Internet-of-Things (IoT) deployments. 
As Long Range (LoRa) is currently one of the most prominent LPWA technologies, we propose in this paper a stochastic-geometry-based framework to analyse the uplink transmission performance of a multi-gateway LoRa network modelled by a Matern Cluster Process (MCP). The proposed model is first to consider all together the multi-cell topology, imperfect spreading factor (SF) orthogonality, random start times, and geometric data arrival rates. Accounting for all of these factors, we initially develop the SF-dependent collision overlap time function for any start time distribution. Then, we analyse the Laplace transforms of intra-cluster and inter-cluster interference, and formulate the uplink transmission success probability. Through simulation results, we highlight the vulnerability of each SF to interference, illustrate the impact of parameters such as the network density, and the power allocation scheme on the network performance. Uniquely, our results shed light on when it is better to activate adaptive power mechanisms, as we show that an SF-based power  allocation that approximates LoRa ADR, negatively impacts nodes near the cluster head. Moreover, we show that the interfering SFs degrading the performance the most depend on the decoding threshold range and the power allocation scheme.
\end{abstract}
\begin{IEEEkeywords}
LoRa, Stochastic Geometry, imperfect SF orthogonality, random start time, collision time overlap, success probability.
\end{IEEEkeywords}

\section{Introduction}
Low-Power Wide-Area (LPWA) Networks (LPWANs) are emerging as a prominent communication solution, addressing the challenging growth, ubiquity, and diversity of the Internet-of-Things (IoT) landscape, while reconciling low-cost and low-energy requirements. LoRa is currently one of the promising solutions among emerging LPWA technologies. LoRa accommodates several tune-able technical parameters like the spreading factor (SF), which specifies the number of bits per symbol, the coding rate (CR), which determines the number of bits used for error correction, transmit power and bandwidth (Bw) \cite{loraCss}. By tuning these parameters, LoRa offers adaptive schemes that can answer different IoT scenarios and applications requirements. 
It is important to understand how such parameters affect performance. 

Indeed several studies have looked into LoRa performance analysis and optimisation \cite{LoraZijin,CanLoRascale, Imperfect_ortho,LoRaThroughput,caillouet:hal-02267218,LoraMulticell,SFAllocation,LoRaSGRain,thresh,LoraenergyHarv}. 
However, most of these studies assume perfect SF-orthogonality and almost exclusively limit their investigations to the impact of interference coming from the nodes using the same SF. From this perspective, a LoRa network can be interpreted as the aggregation of independent sub-networks, each operating in a different SF. 
Under the aforementioned assumption, the performance of a multi-cell LoRa network coexisting with other unlicensed radio technologies was studied in \cite{LoraZijin}. 
In \cite{CanLoRascale}, the scalability analysis of a single LoRa cell was provided. However, the outage condition was formulated based only on the dominant interfering signal.

The assumption of perfect orthogonality has been empirically questioned in \cite{Imperfect_ortho}. Few research studies have, hence, begun to consider non-perfect or quasi-orthogonality use cases. Among these studies, some works explored geometry-less schemes like \cite{LoRaThroughput,caillouet:hal-02267218}, while other works used a geometry-based approach like \cite{LoraMulticell}, which modelled a multi-cell LoRa network using two different cluster processes: Matern Cluster Process and Matern Hard core Process. Besides, in their signal-to-interference and noise ratio (SINR) formulations, most of the works considered co-subchannel rejection thresholds between each two SFs by considering the interference from only one SF-set \cite{LoRaThroughput, LoraMulticell,SFAllocation}. The first concern regarding such thresholds is that they are empirical and hence not unique.
For instance, the values empirically validated in \cite{Imperfect_ortho} and adopted in \cite{LoraMulticell} are different from those in \cite{thresh} which are used in \cite{SFAllocation}. 
The second concern about this approach is whether or not it captures correct decoding methods at the gateway level. Does a LoRa gateway decode the signal using a pairwise-scheme based on the SF value of the interfering packet? For these reasons, we choose to conduct the analysis following a more general approach by varying the range of decoding thresholds and considering interference from all the SFs.

The absence of coordination between nodes in LoRa Aloha-like asynchronous system leads to an interfering power that changes over time. Although essential, especially with the SF-related variable packet's time on-air (ToA), interference time dependence has been generally underestimated and neglected in LoRa network analysis. Only a few studies have integrated it, such as \cite{SFAllocation} \cite{LoRaSGRain}. In \cite{SFAllocation} a collision time probability distribution was formulated based on the difference between uniform start times and used to analyse SF allocation in a single gateway topology assuming the rejection thresholds previously mentioned. In \cite{LoRaSGRain}, a spatiotemporal density was used to study a single gateway under only co-SF interference which resulted in treating LoRa like pure Aloha. 

In this paper, we aim to bridge these research gaps by considering the analysis of LoRa uplink transmissions in a multi-cell topology with imperfect orthogonality between different SFs and random transmission start times. We use stochastic geometry, which is known for its ability to capture different sources of randomness within the network \cite{SG-TUTO}. Our main contributions are summarized as follows:
\begin{itemize}
    \item A Novel spatiotemporal mathematical model is presented for a multi-gateway LoRa network; it accounts for the imperfect SF-orthogonality and the collision overlap time. 
    \item The SF-based collision overlap time function is formulated for random transmission start times.
    
    \item A general analytical expression of the transmission success probability is derived; it can scale down to particular cases and other published works.
    
    \item The vulnerability of SFs to interference is assessed, and their relationship to one another performance is analyzed.
    
   \item The network parameters that impact the success transmission probability, and hence the scalability of the network are studied, including node density, power allocation schemes, and decoding thresholds.
\end{itemize}

\section{System Model}\label{sec2}

In this paper, we use a Matern Cluster Process (MCP) to model a multi-gateway LoRa network. This cluster process allows us to account for the clustered-nature of LoRa, as an operator-free potentially unplanned technology. According to this cluster process, LoRa gateways $L_i$ are distributed following a homogeneous Poisson Point Process (PPP) $\Phi_{G}=\{y_i,i=1,2,\dots \}$ with intensity $\lambda_{G}$, where $y_i\in \mathbb{R}^2$ is the location of the $i$'th LoRa gateway. Each cluster $C_i$ centred is at $L_i$ and has a radius $R$. Within the area of each cluster, LoRa end-devices (EDs) are uniformly scattered around $L_i$ and form a PPP $\Phi_{ED,i}=\{x_{ij},j=1,2,\dots \}$ of intensity $\lambda_{ED}$, where $x_{ij}\in \mathbb{R}^2$ is the location of the $j$'th LoRa ED in the $i$'th  cluster. The overall superposition of $\Phi_{ED,i}$ captures the position of all the children nodes and gives the desired MCP-based network. 

Furthermore, each LoRa ED can be assigned to an SF in $\mathcal{S}$=$\mathbf{\{SF_1, \dots,SF_N\}\ }$, where $N$ is the total number of available SFs. We adopt an equal-interval-based (EIB) SF allocation scheme for which each cluster $C_i$ is divided into $N$ annuli $A_q$ delimited by $d_{q-1}$ and $d_q$, where $q$ $\in$ $\mathcal{Q}$=$\mathbf {\{\ 1,2,..., N\}\ }$ is standing for the $q$'th SF. 
Each annulus $A_q$ is of width $\omega = \frac{R}{N}$ and hence $d_{q-1}=(q-1)\omega$ and $d_q=q\omega$. The average nodes number in each annulus is $N_{q}=\lambda_{ED}\pi (d_{q}^2-d_{q-1}^2)$. 
The overall spatio-temporal model of the network can be interpreted as an independently marked process where the ground process is formed by the nodes positions and the marks represent the transmission start time of each node \cite{SG-ALOHA}. 
The time marks are independent since the medium access technique used by LoRa is un-slotted Aloha-like where nodes send their packets independently without any prior coordination or synchronization. At each device, the packets are generated according to a geometric distribution with parameter $a \in [0, 1]$. By virtue of the independent thinning of a homogeneous PPP \cite{SG-book}, the subset of transmitting nodes form a homogeneous PPP $\tilde{\Phi}_{ED,i}$ of intensity a$\lambda_{ED}$ (See Fig.\ref{Fig:Network}).
We consider a power-law path-loss propagation model where the signal attenuates with the propagation distance at the rate $r^{-\eta}$, $\eta > 2$ is the path-loss exponent. Added to the large scale fading, we have Rayleigh block fading channels with unit mean exponentially distributed channel gains $g_{ij}$, i.e. $g_{ij} \sim exp(1)$. All the channels are assumed to be independent of the space and time dimensions.
\graphicspath{ {Sections/Plots/} }
\begin{figure}[!h]
\centering

\includegraphics[scale=.35]{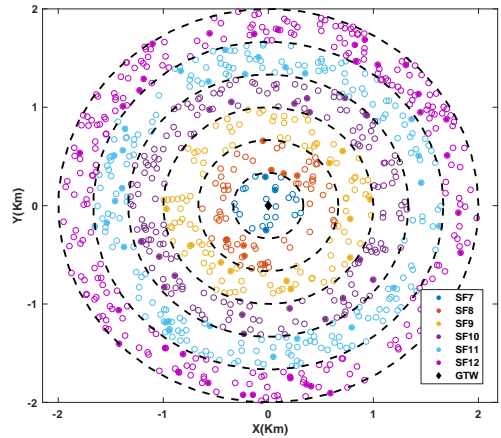}
%\vspace*{-10mm}
\caption{Example of EIB SF allocation in a single-gateway LoRa network with active (filled dots) and inactive (empty dots) nodes, with $a=0.1$, $\lambda_{ED}=80$ nodes/Km$^2$, and $R=2$ Km.}
\label {Fig:Network}
\end{figure}
%\vspace*{-10mm}

%---------------------------------------------------------------------------------%
\section{Stochastic Geometry Analysis }\label{sec3}
The received Signal to Interference and Noise Ratio (SINR) at the typical LoRa receiver from a typical LoRa node, located at $r_{0}= \lVert x_{00} \rVert$ and emitting with SF $q_0\in \mathcal{Q}$, is formulated as:
    \begin{equation}\label{SINR}
    {SINR(r_{0},q_{0}) = \frac {P_t g_{00} \alpha r_{0}^{-\eta}}{{I_{intra}}+ I_{inter} + \sigma^2}},
    \end{equation}
    where $I_{intra}$ is the intra-cluster interference coming from active nodes within the same cluster, $I_{inter}$ is the inter-cluster interference originating from transmitting nodes in other clusters, and $\sigma^2$ is the variance of the additive white Gaussian noise (AWGN). $I_{intra}$ and $I_{inter}$ account for interference from the same SF (Co-SF) and from different SFs (Inter-SF).

%---------------------------------------------------------------------------
\subsection{SF-Dependent Collision Overlap Time}   
%---------------------------------------------------

As LoRa uses interleaving and repetition codes, we consider an averaging over the exchanged packet duration to account for the time dependence of the interference \cite{SG-ALOHA}. In contrast to ordinary Aloha models, LoRa has a variable packet duration $l_{q}$ since the packet Time-On-Air (ToA) is linked to the SF used in the transmission. The variable time-on-air leads to an SF-dependent collision overlap time. 

We consider a typical LoRa node located at $x_{00} \in \tilde{\Phi}_{ED,0}$ emitting with SF $q_0\in \mathcal{Q}$ and communicating with a typical gateway placed at the origin. Without loss of generality, we assume that the typical LoRa node starts its desired transmission with $SF=q_{0}$ at $T_{0}=0$. The time-averaged $I_{intra}$ and $I_{inter}$ interference experienced by the receiver are given by:
    \begin{align}
    \label{MeanInterf}
    I_{intra}&=I^{mean}_{intra} =\frac{1}{l_{q_0}} \int_{T_{0}}^{T_{0}+l_{q_0}}I_{intra}(t)dt\nonumber\\
    &=\sum_{q \in {Q}} \sum_{x\in\tilde{\Phi}_{ED,0} \backslash{x_{00}} } \mathds{1}_{j,q} P_{q}\alpha \lVert x_{0j} \rVert^{-\eta} h_{q_{0},q}(T_{ij})g_{0j},\\
    I_{inter}&=I^{mean}_{inter} =\frac{1}{l_{q_0}} \int_{T_{0}}^{T_{0}+l_{q_0}}I_{inter}(t))dt\nonumber\\&=
    \sum_{q \in {Q}} \sum_{ y \in{\Phi_{G}\backslash{y_0}}}\sum_{ x\in\ {\tilde{\Phi}_{ED,i}}}\mathds{1}_{j,q}P_{q}\alpha\lVert y_{i}+x_{ij} \rVert^{-\eta} h_{q_{0},q}(T_{ij})g_{ij},\label{MeanInterf1}
    \end{align}
    where $\mathds{1}_{j,q}$ is the indicator function of $ED_{j}$ transmitting at SF ${q}$ and  $h_{q_{0},q}(T_{ij})$ is the collision overlap time function between the LoRa node located at $x_{ij}\in \tilde{\Phi}_{ED,i}$ with random transmission start time $T_{ij}$ and the typical user. $h_{q_{0},q}(T_{ij})$ is expressed as:
    \begin{align} 
    h_{q_{0},q}(T_{ij})&= \frac{1}{l_{q_0}} \int_{T_{0}}^{T_{0}+l_{q_0}}{ \mathbb{1}\big ( \; x_{ij}\textrm{ overlaps with} \; x_{00} \big )}\quad d(t),
    \end{align}
 
Because of duty cycle restriction where a node is active only for \%1 and since $l_{6}< 100\times l_{1}$, a desired packet will not be interfering with a first and second transmissions from the same node. Assuming all the active nodes (except the typical user) start transmitting randomly in a contention window  $[-T_{c},T_{c}]$.
    The collision overlap time is expressed in the following lemma:
    \begin{lemma}\label{lemcol}
    The collision overlap time function $h_{q_0,q}(T_{ij})$ between the desired node $x_{00}$ and the interfering node $x_{ij}$ transmitting with SF $q$ at random time $T_{ij}$ is
        \begin{align}
            h_{q_0,q}(T_{ij})&= 
                \begin{cases}
                \frac{l_{q_0}-T_{ij}}{l_{q_0}}, \quad \text{ if } \left(l_{q_0}-l_{q}\right)^+ \leq T_{ij} \leq l_{q_0},\\
                \frac{\min\left( l_{q_0},l_{q}\right)}{l_{q_0}}, \quad \text{ if } -\left(l_{q}-l_{q_0}\right)^+\leq T_{ij}\leq  \left(l_{q_0}-l_{q}\right)^+,\\
                \frac{l_{q}+T_{ij}}{l_{q0}}, \quad \text{ if } -l_{q} \leq T_{ij} \leq -\left(l_{q}-l_{q_0}\right)^+,\\
                0, \quad \text{if } -T_{c} \leq T_{ij} < -l_{q}\text{ or } l_{q_0} <T_{ij}\leq T_{c},
                \end{cases}
        \end{align}
        where $\left(t \right)^+ = max(t,0)$.
   \end{lemma}
    
%---------------------------------------------------------------------------%
    \begin{proof}
    The proof is in Appendix \ref{appencol}.
    \end{proof} 
%--------------------------------------------------------------------------%
    \begin{corollary}\label{corrEcol}
    Using Lemma \ref{lemcol}, and assuming that the transmission starting time of the interfering nodes is uniformly distributed between $[-T_{c},T_{c}]$, $T_{ij}\overset{Dist}{=}U(-T_{c},T_{c})$, we show that
        \begin{align}
            \mathbb{E} _{T_{ij}}\left[\frac{1}{ 1+ u h_{q_0,q}\left(T_{ij}\right) } \right]
            &= 1 - \frac{l_{q_0}+l_{q}}{2T_c} + \frac{ l_{q_{0}} }{T_{c} u} \log\left( u \frac{ min\left(l_{q},l_{q_0} \right) }{l_{q_0} } + 1 \right)\notag\\ 
            &+ \frac{\left\vert l_{q_0}-l_{q} \right\vert }{ 2 T_c \left( 1 +  u\frac{ min\left(l_{q},l_{q_0}\right) }{ l_{q_0} } \right)},\label{eqcorrEcol}
        \end{align}
        where $\mathbb{E}_{T_{ij}}[\cdot]$ is the expectation operator with respect to $T_{ij}$.
  \end{corollary}
%--------------------------------------------------------------------------%
    \begin{proof}
    The proof is in Appendix \ref{appencorrEcol}.
    \end{proof}
%----------------------------------------------------------------------------%
%----------------------------------------------------------------------------
\subsection{Transmission Success Probability}%\label{sec3}
%-----------------------------------------------------------------------------
The typical LoRa gateway is able to receive and successfully decode the desired signal if its instantaneous SINR surpasses a reference decoding threshold $\gamma_{th}$ as

    \begin{align}%\label{eq1}
        P_{Succ}(r_{0},q_{0})&= \prob \{ SINR (r_{0},q_{0})\geq \gamma_{th}\}
                \nonumber\\
                &\overset{(a)}{=} e^{-\rho\sigma^2}\mathbb{E}_{I_{intra}}\{e^{-\rho I_{intra}}\}  \mathbb{E}_{I_{inter}}\{e^{-\rho I_{inter}}\}\nonumber\\
            &=e^{-\rho\sigma^2}\mathfrak{L}_{I_{intra}}(\rho)\mathfrak{L}_{I_{inter}}(\rho),
          \end{align}
   where $\rho=\frac{\gamma_{th} r_{0}^{\eta}}{P_{q_{0}}\alpha}$, (a) was obtained using the exponential distribution of the channel $g_{00}$, and $\mathfrak{L}_{I_{intra}}(\cdot)$ and $\mathfrak{L}_{I_{inter}}(\cdot)$ are the Laplace transforms of $I_{intra}$ and $I_{inter}$, respectively.
 
%---------------------------------------------------------------------------------% 
   
     In order to derive the expression of the success probability, we need first to investigate the expressions of the Laplace transforms of $I_{intra}$ and $I_{inter}$.
  \begin{theorem}\label{LIintra}
    The Laplace transform of $I_{intra}$ is given by:
        \begin{align}
        \label{EqIntra}
        \mathfrak{L}_{I_{intra}}(\rho) &\approx \prod_{q \in {Q}} \exp\Bigg(- 2\pi a\lambda_{ED}\left(I_{1}(q)-I_{2}(q)-I_{3}(q)\right)\Bigg),
        \end{align} 
        with $I_{1}(q)=\frac{ \left(l_{q_{0}}+l_{q}\right)}{4T_{c}}\left(d_{q}^2-d_{q-1}^2\right)$, 

            \begin{align}
            I_{2}(q) &=\frac{min(l_{q},l_{q_{0}})}{2T_{c}b{ (\eta+2)}}\Biggl[ d_{q}^2 \Bigg(\eta b \, _2F_1\left(1,-\frac{2}{\eta};\frac{\eta-2}{\eta};-b d_{q}^{-\eta}\right)\notag\\
            &+2 d_{q}^\eta \log \left(b d_{q}^{-\eta}+1\right)\Bigg) -d_{q-1}^2 \Bigg(\eta b \, _2F_1\left(1,-\frac{2}{\eta};\frac{\eta-2}{\eta};-b d_{q-1}^{-\eta}\right)\notag\\&+2 d_{q-1}^\eta \log \left(b d_{q-1}^{-\eta}+1\right)\Bigg)\Biggr],\\
            I_{3}(q)&=\frac{\left|l_{q_{0}}-l_{q_{j}}\right|}{2T_{c}b(\eta+2)}\Biggl[d_{q}^{\eta+2}\,
            _2F_1\left(1,\frac{\eta+2}{\eta};\frac{\eta+2}{\eta}+1;-\frac{d_{q}^\eta}{b}\right)\\
            &-d_{q-1}^{\eta+2}\,
            _2F_1\left(1,\frac{\eta+2}{\eta};\frac{\eta+2}{\eta}+1;-\frac{d_{q-1}^\eta}{b}\right)\Biggr],\nonumber
             \end{align}
            where $_2 F_{1}(\cdot)$ is the Gaussian hypergeometric function \cite{integralsbook}, and $b=
                \frac{P_{q}}{P_{q_{0}}}\frac{ min\left( l_{q},l_{q_{0}} \right) }{l_{q_{0}}} \gamma_{th} r_{0}^{\eta}= \alpha P_q \frac{ min\left( l_{q},l_{q_{0}} \right) }{l_{q_{0}}} \rho $.
 \end{theorem}

%---------------------------------------------------------------------------------%
    \begin{proof}
    The proof is in Appendix \ref{appenIintra}.
    \end{proof}
%---------------------------------------------------------------------------------%

%---------------------------------------------------------------------------------%
\begin{theorem}\label{LIinter}
    The Laplace transform of $I_{inter}$ is given by:
    \begin{align}
    \mathfrak{L}_{Inter} (\rho) 
        &\approx \prod_{q \in Q} exp\Bigg[-2\pi \lambda_{G}a N_{q}\frac{\pi}{\eta\sin(\pi\frac{2}{\eta})}
        \left(\alpha {P_{q}\rho }\right)^{\frac{2}{\eta}}\notag \times \frac{1}{2Tc}\\& \left(\frac{2\eta}{\eta+2}l_{q_{0}}
       \left(min(1,\frac{l_{q}}{l_{q_{0}}})\right)^{\frac{\eta+2}{\eta}}+\left(min(1,\frac{l_{q}}{l_{q_{0}}})\right)^{{\frac{2}{\eta}}}\vert l_{q_0}-l_{q} \vert\right) \Bigg].\label{EqInter}
    \end{align}
 \end{theorem}

 \begin{proof}
The proof is in Appendix \ref{appenIinter}.
\end{proof}
Given (\ref{EqIntra}) and (\ref{EqInter}), the general expression of  $P_{Succ} $ is given in (\ref{psucc}).
%-----------------------------------  

%-------------------------------------------------------------------
%%%

	\newcounter{mytempeqncnt4}\begin{figure*}[!t]% 
	\normalsize%
	\setcounter{mytempeqncnt4}{\value{equation}}
    	\begin{align}\label{psucc}
    	    P_{Succc}\left(r_0,q_0\right) &= e^{-\rho\sigma^2}\prod_{q \in {Q}} \Bigg(e^{- 2\pi a\lambda_{ED}\left(I_{1}(q)-I_{2}(q)-I_{3}(q)\right)}
    	    e^{\frac{-2\pi^2 \lambda_{G}a N_{q}}{\eta\sin(\pi\frac{2}{\eta})}
        \left(\alpha {P_{q}\rho }\right)^{\frac{2}{\eta}} \left(\frac{1}{2Tc}\frac{2\eta}{\eta+2}l_{q_{0}}\left(min(1,\frac{l_{q}}{l_{q_{0}}})\right)^{\frac{\eta+2}{\eta}}+\left(min(1,\frac{l_{q}}{l_{q_{0}}})\right)^{{\frac{2}{\eta}}}\vert l_{q_0}-l_{q} \vert\right) }\Bigg).
    	\end{align}
	\hrulefill%
	\vspace*{2pt}
	\end{figure*}

\subsection{Special Cases}
Here, we state few special cases deduced from our analytical results that can scale down to other published works:
\begin{itemize}
    \item[](i) {Perfect Orthogonality:}
If we consider perfect orthogonality, the transmission success probability simplifies to:

$P_{Succ}(r_{0},q_{0})=e^{-\rho\sigma^2}e^{- 2\pi a\lambda_{ED}\left(I_{1}(q_{0})-I_{2}(q_{0})\right)}e^{\frac{-2\pi^2 \lambda_{G}a N_{q_{0}}l_{q_{0}}r_{0}^{2}\gamma_{th}^{\frac{2}{\eta}}}{Tc (\eta+2) \sin(\pi\frac{2}{\eta})}}.$
    \item[](ii) {Single Gateway Toplogy:} For a single gateway topology, only intra-cluster interference are considered:\\ $P_{Succ}(r_{0},q_{0})=e^{-\rho\sigma^2} \prod_{q \in {Q}} e^{(- 2\pi a\lambda_{ED}\left(I_{1}(q)-I_{2}(q)-I_{3}(q)\right))}$. 
    \item[](iii) {Only one interfering SF:} Considering the $q$'th SF:
    
$P_{Succ}(r_{0},q_{0})=e^{-\rho\sigma^2}e^{(- 2\pi a\lambda_{ED}\left(I_{1}(q)-I_{2}(q)-I_{3}(q)\right))}\times \\e^{\frac{-2\pi^2 \lambda_{G}a N_{q}}{\eta\sin(\pi\frac{2}{\eta})}
        \left(\alpha {P_{q}\rho }\right)^{\frac{2}{\eta}} \left(\frac{1}{2Tc}\frac{2\eta}{\eta+2}l_{q_{0}}\left(min(1,\frac{l_{q}}{l_{q_{0}}})\right)^{\frac{\eta+2}{\eta}}+\left(min(1,\frac{l_{q}}{l_{q_{0}}})\right)^{{\frac{2}{\eta}}}\vert l_{q_0}-l_{q} \vert\right) }$. 
    \item[](iv) {Same Power Allocation:} Assuming all the SFs use the same power, $P_{q}=P_{q_{0}}$, and $b$ in (\ref{EqIntra}) simplifies to $\frac{ min\left( l_{q},l_{q_{0}} \right) }{l_{q_{0}}} \gamma_{th} r_{0}^{\eta}$.
\end{itemize}

\section{Simulation Results}\label{sec4}

In this section, we validate our analytical model using Monte Carlo (MC) simulations. The packet size is fixed to $25$ bytes. 
The packet time-on-air depends on the used SF and is calculated, based on each SF Data Rate \cite{CanLoRascale} ($l_{1}=0.036$s, $l_{2}= 0.064$s, $l_{3}=0.113$s, $l_{4}= 0.204$s, $l_{5}=0.365$s, and $l_{6}=0.682$s).
LoRa coverage radius for dense urban environment is $2$ Km and a typical metropolitan area of $100$ km$^2$ can be covered by $30$ gateways \cite{Sim_Param-Ref}. Hence, in our simulation scenario we assumed R$=2$km and $\lambda_{G}$(/Km$^2$)$=0.3$. The bandwidth and the frequency are chosen according to LoRa regulations for the European region: Bw$=125$ KHz and f$_{c}=868$ MHz , the contention window is $Tc=1.5$ Sec. Unless otherwise mentioned, the parameters used in the simulations are: $\eta=3$, a$=0.1$, $\lambda_{ED}=100$ Nodes$/$Km$^2$ and $P_{q}=14$ dBm. To analyze the impact of power allocation on the performance, we tested two schemes: same power allocation and SF-based power allocation. For the first scheme,  $P_{q}=14$ dBm $\forall$ q; while for the second, the power is attributed according to the used SF (Higher SFs are assigned higher powers) which is close to the way LoRa Adaptive Data Rate (ADR) works \cite{LoraSpec} ($P_{1}=2$ dBm, $P_{2}= 5$ dBm, $P_{3}=8$ dBm, $P_{4}= 11$ dBm, $P_{5}=14$ dBm, $P_{6}=20$ dBm). 
To calculate the performance metric, LoRa nodes are deployed according to a MCP and kept fixed for the simulation setup which is similar to real deployment scenarios in most smart city IoT applications. The desired node position is fixed based on the SF to investigate at $r_{0}(q)=d_{q-1}+\frac{\omega}{2}$ and its transmission status remains equal to 1 (always active). At each simulation step, the interfering nodes are determined based on their data status which follows a geometric distribution; once they have data to transmit, the transmission start time of each node is randomly generated following a uniform distribution. The collision overlap time with the desired packet is then calculated and multiplied by the interfering power. For MC simulations, the transmission success probability of each SF, under both perfect/imperfect SF orthogonality, is found by averaging over the number of simulations. In all the figures of this section, markers illustrate results obtained by MC simulation.

\graphicspath{ {Sections/Plots/} }
\begin{figure}[t]
    \begin{subfigure}[b]{0.47\columnwidth}
    \includegraphics[scale=0.14]{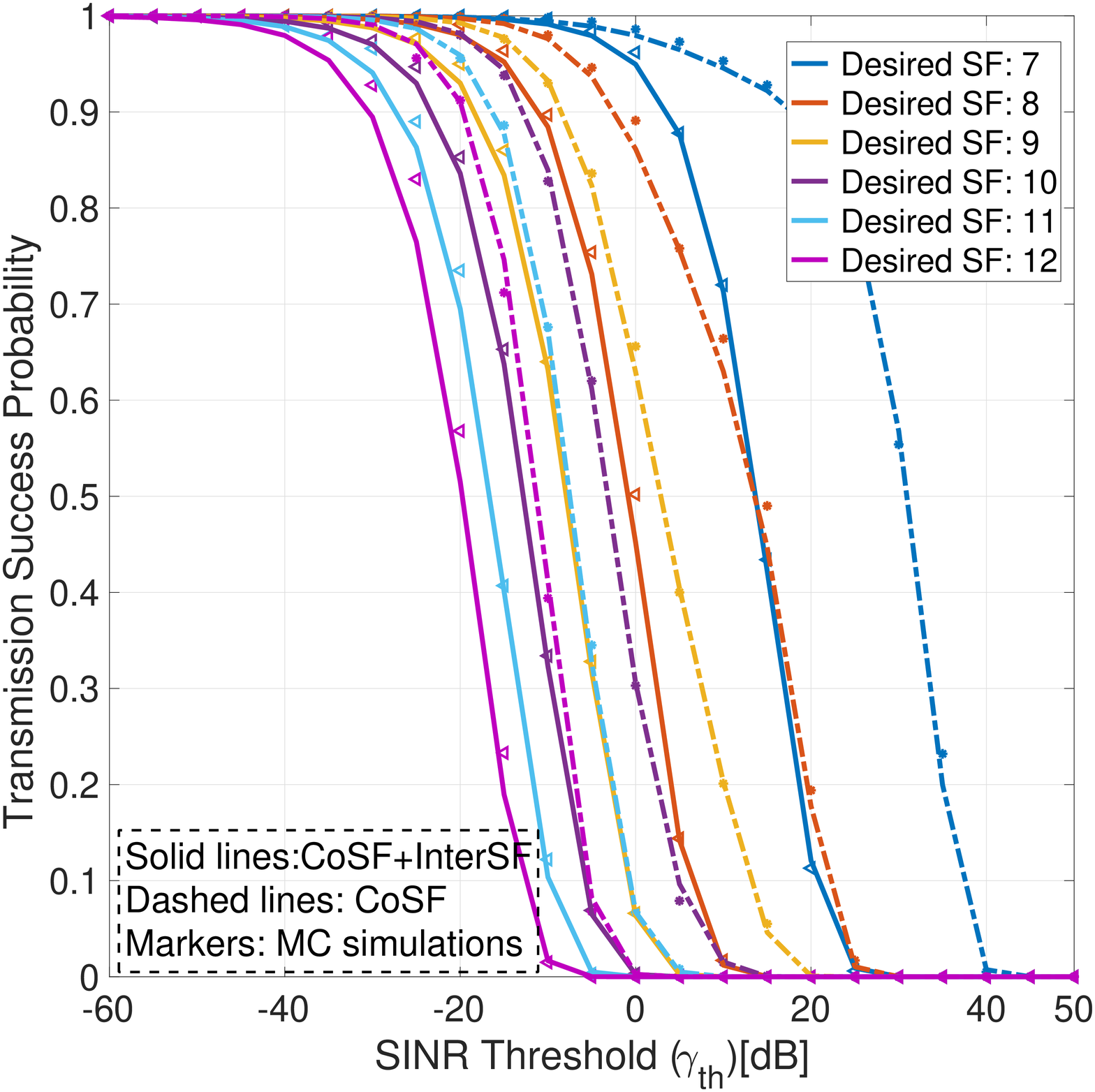}
    \label{Intra}
       \vspace{-3mm}
    \caption{Single Gateway topology}
    \end{subfigure}\hspace{2mm}
    \begin{subfigure}[b]{0.47\columnwidth}
    \includegraphics[scale=0.14]{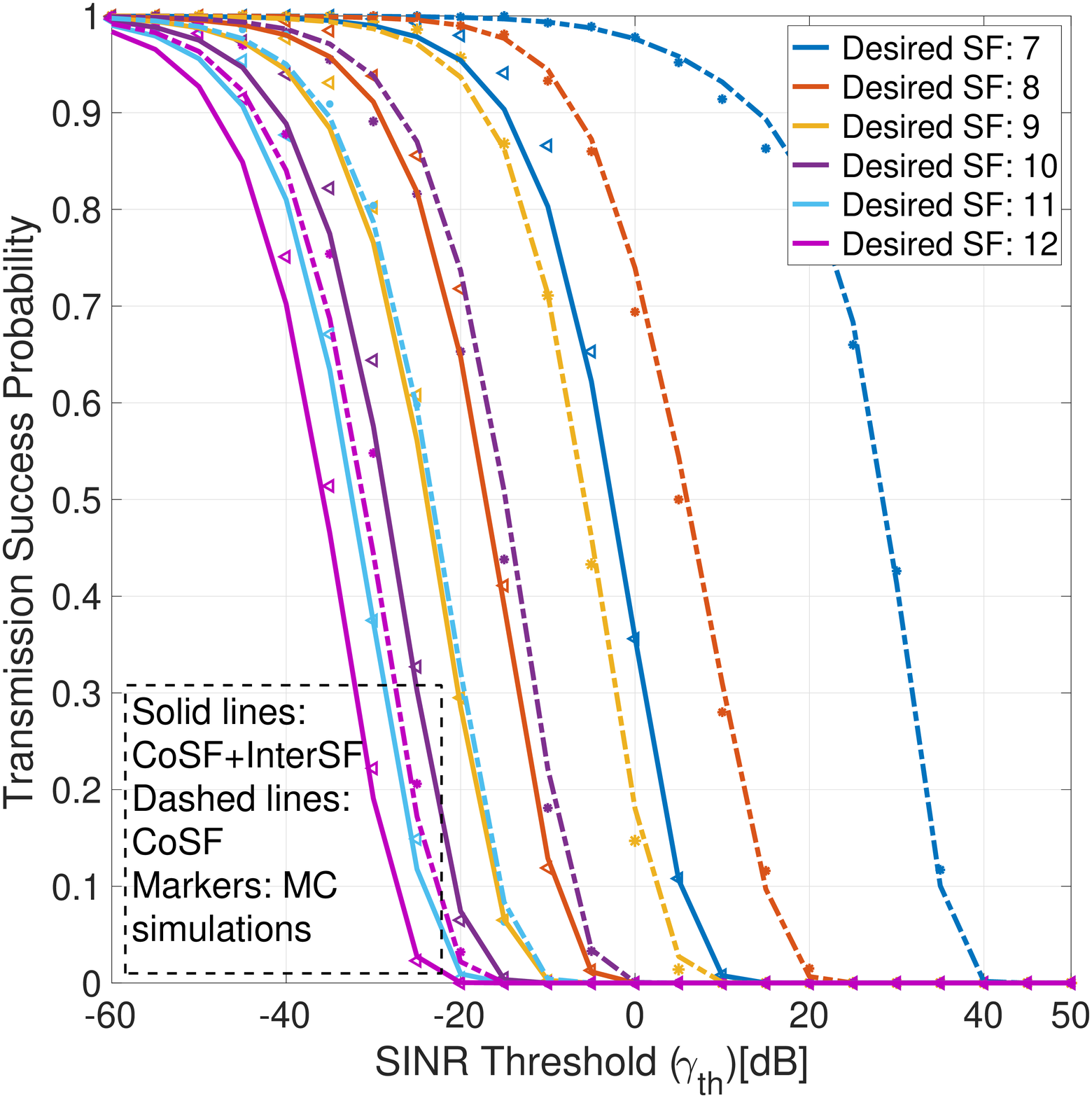}
    \label{InterIntra}
       \vspace{-3mm}
    \caption{Multiple Gateway topology}
    \end{subfigure}
    \caption{Transmission success probability ($P_{Succ}$) versus different SINR thresholds ($\gamma_{th}$).}
    \label{DiffSF}
    \vspace{-3mm}
\end{figure}
\begin{figure}[!htb]
    \centering
    \includegraphics[scale=0.22]{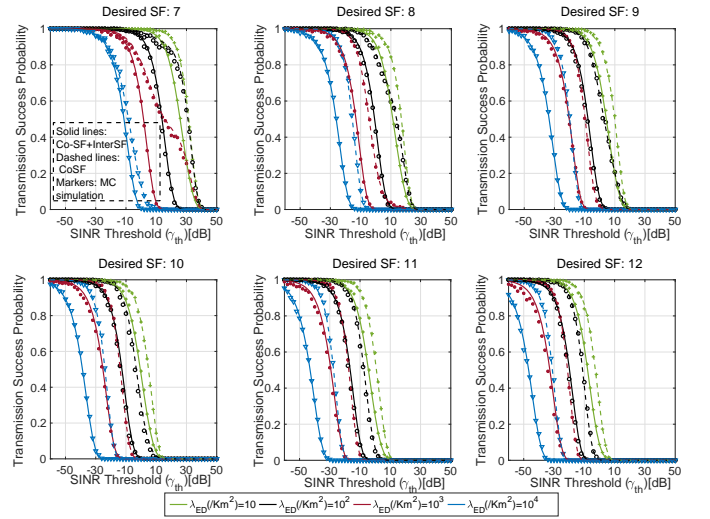}
    %\vspace*{-8mm}
    \caption{Transmission success probability versus SINR thresholds for different $\lambda_{ED}$ in a single LoRa cell.}
    \label{Impact_Dens}
\end{figure}

\graphicspath{ {Sections/Plots/} }
\begin{figure}[t]
    \centering
    %scale=0.5
    \includegraphics[scale=0.218]{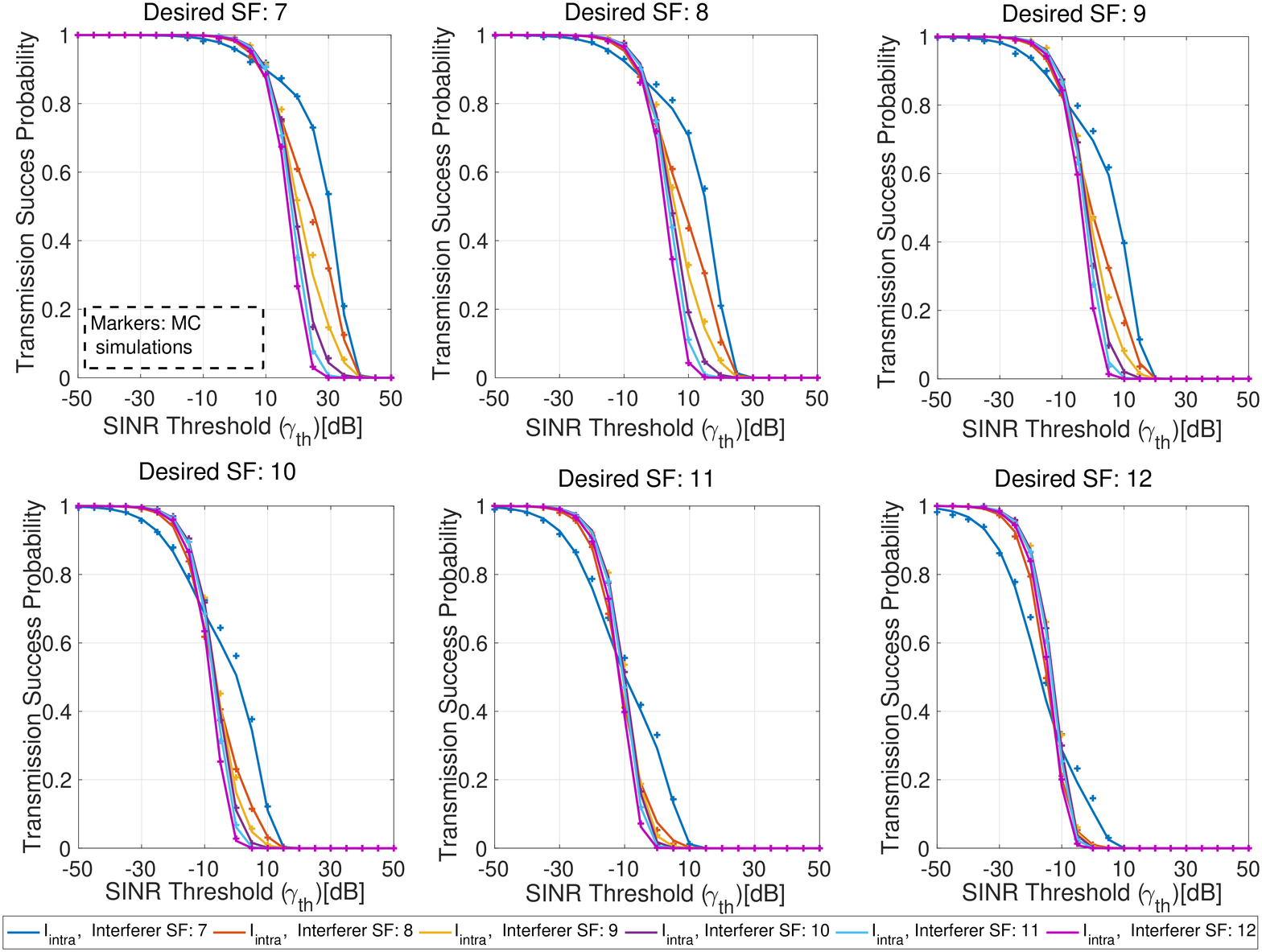}
    \caption{Transmission success probability versus SINR thresholds under Per-SF interference and same power allocation in a single LoRa cell with $\lambda_{ED}=200$ Nodes/Km$^2$.}
    \label{PerSFSamePower}
    \vspace*{-5mm}
\end{figure}

\begin{figure}[t]
    \centering
    %scale=0.5
    %\hspace{-29mm}
    \includegraphics[scale=0.218]{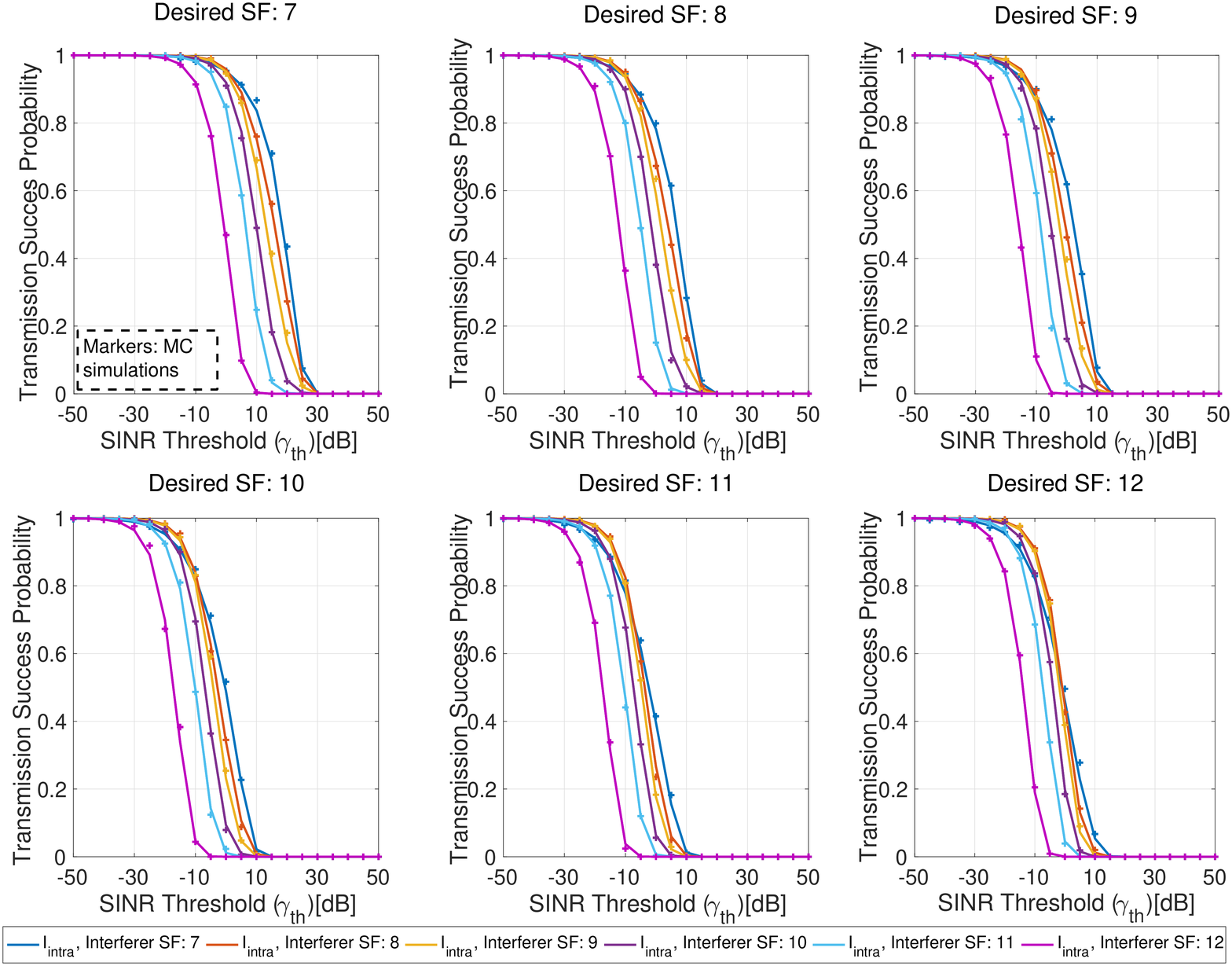}
    %\vspace*{-5mm}
    \caption{Transmission success probability versus SINR thresholds under Per-SF interference and SF-based power allocation in a single LoRa cell with $\lambda_{ED}=200$ Nodes/Km$^2$.}
    \label{PerSFDiffPower}
    \vspace*{-5mm}
\end{figure}

Fig.\ref{DiffSF} shows the transmission success probability of each desired SF versus different SINR thresholds for both single gateway topology and multi-gateway topology. 
Solid lines illustrate the impact of both interference types (Co-SF and Inter-SF), while dashed lines illustrate the impact of only Co-SF interference. We can see that for each SF the probability of successful transmission under aggregated interference from different SFs is considerably lower than the result obtained by considering only Co-SF interference.
Hence, we can say that the perfect orthogonality assumption commonly used results in an overestimation of the network performance which may impact the network dimensioning and planning. We can see also that, as expected, packet transmission success decreases when SF increases. This can be explained by the fact that higher SFs have longer time on air which leads to longer time overlap with the desired packet and hence higher interference exposure.
 
In Fig.\ref{Impact_Dens}, we plotted the transmission success probability of different desired SFs under both aggregated interference and only co-SF interference, for different devices densities in a LoRa gateway. Higher nodes densities degrade all the SFs performance and its impact is much important for higher SFs as it appears for SF$=11$ and SF$=12$.
 To assess the impact of SFs on one another performance, we plotted in Fig.\ref{PerSFSamePower} and Fig.\ref{PerSFDiffPower} the transmission success probability of each SF under interference from one specific interfering SF-set. 
 In Fig.\ref{PerSFSamePower}, we considered the same power allocated to all the nodes independently from the used SF; while in Fig.\ref{PerSFDiffPower}, we used the SF-based power allocation scheme which is closer to the way LoRa ADR works. 
 An examination of these figures reveals that for the case of same power allocation, at lower SINR thresholds, lower SFs tend to have the worst impact on the success probability while at higher SINR thresholds, higher SFs have the worst impact. This observation stresses the importance of SF allocation on the network performance.
 Under SF-based power allocation, only one common behaviour for all the SFs is recognized: higher SFs decrease transmission success probability more, independently of the SINR threshold. This observation is more aligned with the commonly believed fact that the higher SFs induce more interference as they stay  active for longer in the network. 
 Moreover, the SF-based power allocation decreases the performance of lower SFs and improves the performance of higher SFs, compared to the same power allocation. This shows that the preference of same or SF-based power allocation scheme depends on the SF of the desired node. 
 For instance, in the case of desired SF$=12$, for an SINR threshold equal to $-10$dB, the success probability is around $30\%$ for all interfering SFs, whereas under the SF-based power allocation the success probability overcomes $70\%$ for SFs lower than $11$. 
 
 Figs.\ref{PerSFSamePower}, \ref{PerSFDiffPower} and \ref{Impact_Dens} confirm that decoding thresholds are not fixed and do not depend only on the SF of desired and interfering nodes, they are also impacted by the nodes density and the power allocation. These thresholds decrease when the interfering SF increase and when the nodes density becomes higher.

\section{Conclusion}

Using stochastic geometry, we analysed the transmission success probability of a multi-gateway LoRa-based LPWA network. We demonstrated that limiting the analysis to the impact of Co-SF interference specifically can lead to an overestimation of the network performance. The incorporation of time dimension with the formulated collision overlap function better depicts the interference temporal dynamic and makes, as a result, the analysis more realistic. We also showed that power allocation schemes play a significant role in the vulnerability of each SF to interference and that decoding thresholds depend on several network parameters. Uniquely, our results suggest that activating an adaptive power allocation schemes like LoRa ADR would be advantageous for nodes far from the gateway more than other nodes. This observation suggests a potential future work to validate this behaviour in a real deployment scenario.

\bibliographystyle{IEEEtran}
\bibliography{Bibl}

\appendices
%------------------------------------------------------------%---------------------------------------------------------------------------------%
\section{Proof of Lemma \ref{lemcol}}\label{appencol}
%---------------------------------------------------------------------------------%

Assuming an interfering node located at $x_{ij}$ that starts transmitting its packet of duration $l_{q}$ at instant $T_{ij}$ randomly in a contention window $[-T_{c},T_{c}]$. Then, the collision time overlap between $x_j$ and the typical user is given by:
\begin{itemize}
    \item If $l_{q}\leq l_{q_{0}}$,

         \[ h_{q_{0},q}(T_{ij})=
        \begin{cases}
            \frac{l_{q_{0}}-T_{ij}}{l_{q_{0}}},\quad\text{if } l_{q_{0}}-l_{q}\leq T_{ij} \leq l_{q_{0}},\\
            \frac{l_{q}}{l_{q_{0}}},\quad\text{if } 0\leq  T_{ij} \leq l_{q_{0}}-l_{q},\\ 
            \frac{l_{q}+T_{ij}}{l_{q_{0}}},\quad\text{if }  -l_{q} \leq T_{ij} \leq 0 ,\\ 
            0 \quad \text{if} -T_{c} \leq T_{ij} < -l_{q} \quad \text{or} \quad   l_{q_{0}} < T_{ij}\leq T_{c}.
        \end{cases}\] 
    \item If $l_{q}> l_{q_{0}}$,

        \[ h_{q_{0},q}(T_{ij})=
        \begin{cases}
            \frac{l_{q_{0}}-T_{ij}}{l_{q_{0}}},\quad\text{if } 0\leq T_{ij} \leq l_{q_{0}},\\
            1,\quad\text{if } l_{q_{0}}-l_{q}\leq  T_{ij} \leq 0 ,\\ 
            \frac{l_{q}+T_{ij}}{l_{q_{0}}},\quad\text{if }  -l_{q} \leq T_{ij} \leq l_{q_{0}}-l_{q},\\ 
            0 \quad \text{if}\quad   -T_{c} \leq T_{ij} < -l_{q} \quad \text{or} \quad   l_{q_{0}} < T_{ij}\leq T_{c}.
        \end{cases}\]
    \end{itemize}

%---------------------------------------------------------------------------------%
\section{Proof of Corollary \ref{corrEcol}}\label{appencorrEcol}
%---------------------------------------------------------------------------------%
Assuming that the transmission start time of LoRa active nodes follows a uniform distribution $T_{ij}\overset{Dist}{=}U(-T_{c},T_{c})$ and using Lemma 1, we evaluate $\mathbb{E} _{T_{ij}}\left[\frac{1}{ 1+ s h_{q_0,q}\left(T_{ij}\right) } \right]$ for $l_q\leq l_{q_0}$ as 
            \begin{align}
            &\mathbb{E}_{t}\left[ \frac{1}{1+u h_{q_0,q}(T_{ij}) } \right] 
                =\frac{1}{2T_c}\Bigg[ \int_{-T_c}^{-l_{q}} dt 
                + \int_{-l_{q}}^{0} \frac{1}{1+ u \frac{t+l_{q}}{l_{q_0}}} dt\notag\\\
                &+\int_{0}^{l_{q_0}-l_{q}} \frac{1}{1+u\frac{l_{q}}{l_{q_0}}} dt +\int_{l_{q_0}-l_{q}}^{l_{q_0}} \frac{1}{1+u\frac{l_{q_0}-t}{l_{q_0}}}  dt + \int_{l_{q_0}}^{T_c} dt\Bigg] \notag\\
                &=1-\frac{l_{q_0}+l_{q}}{2T_c} + \frac{l_{q_0}}{T_c u} \log\left( \frac{u l_{q}}{l_{q_0}} + 1 \right)+\frac{l_{q_0}-l_{q}}{2T_c \left(1+u\frac{l_{q}}{l_{q_0}} \right)}.\label{Ecorr0Q1}
            \end{align}
          
        %\item 
        Similarly, we show for $l_{q}> l_{q_0}$, that
            \begin{align}
            \mathbb{E}_{t}\left[\frac{1}{1+u h_{q_0,q}(T_{ij}) }\right] 
        &=1-\frac{l_{q_0}+l_{q}}{2T_c} + \frac{l_{q_0}\log(u+1)}{T_c u}  -\frac{l_{q_0}-l_{q}}{2T_c(1+u)}.%\label{Ecorr0Q2}
            \end{align}

%---------------------------------------------------------------------------------%
\section{Proof of Theorem \ref{LIintra}}\label{appenIintra}
%---------------------------------------------------------------------------------%
 %\overset{(a)}
% \overset{(b)}{
    \begin{align}
        \mathfrak{L}_{I_{intra}}\{s\} 
              & = \mathbb{E}_{x,G,t} \left[e^{-s\sum_{q \in Q} \sum_{x_j \in \tilde{\Phi}_{ED,0,q}\backslash{x_0}} P_q h_{q_{0},q}(T_{ij})\alpha\lVert x_{0,j}\lVert^{-\eta}g_{0j}}\right]
            \\
            & =\mathbb{E}_{x,G,t} \Bigg[\prod_{q \in Q} \prod_{x_j \in\tilde{\Phi}_{ED,0,q}\backslash{x_0}}e^{-s P_q h_{q_{0},q}(T_{ij})\alpha\lVert x_{0,j}\rVert^{-\eta}g_{0j}}\Bigg]\notag\\
            & \overset{(a)}{=} \mathbb{E}_{x,t} \Bigg[\prod_{q \in Q} \prod_{x_j \in\tilde{\Phi}_{ED,0,q}\backslash{x_0}} \frac{1}{1+sP_q h_{q_{0},q}(T_{ij})\alpha\lVert x_{0,j}\rVert^{-\eta}} \Bigg]\notag\\
         &\overset{(b)}{\approx}  \prod_{q\in {Q}}\Bigg[\mathbb{E} _{x,t} \prod_{x_j \in\tilde{\Phi}_{ED,0,q}\backslash{x_0}} \frac{1}{1+sP_{q} h_{q_{0},q}(T_{ij})\alpha\lVert x_{0,j}\rVert^{-\eta}} \Bigg]\notag\\
        &\overset{(c)}{=}\prod_{q \in Q}\Bigg[e^{-2\pi a\lambda_{ED} \int_{d_{q-1}}^{d_{q}}\big(1-\mathbb{E} _{t}\big [\frac{1}{1+s P_q h_{q_{0},q}(T_{ij}) \alpha r^{-\eta}}\big] r dr } \Bigg],\notag
    \end{align}
  where (a) is explained by the independence of channel gains from both the spatial and temporal dimensions and is obtained using the moment generating function (MGF) of the exponential distribution with mean $1$, (b) is an approximation obtained using Fortuin–Kasteleyn–Ginibre (FKG) inequality for $T_{ij}>0$ and extended to $\forall$ $T_{ij}$ through validation by MC simulations, and (c) is obtained by applying the probability generating function (PGFL) of $\tilde{\Phi}_{ED,i}$ and the change of integration coordinates from Cartesian to polar. Using Corollary \ref{corrEcol}, we obtain (\ref{EqIntra}).

\section{Proof of Theorem \ref{LIinter}}\label{appenIinter}

 \begin{align}
        &\mathfrak{L}_{I_{inter}}\{s\} \\ 
           & \overset{(a)} =\mathbb{E}_{y,x,t} \Bigg[\prod_{q \in Q}\prod_{y_i \in\Phi_{G}}\prod_{x_j \in\tilde{\Phi}_{ED,i,q}} \frac{1}{1+sP_qh_{q_{0},q}(T_{ij})\alpha\lVert {y_{i}+x_{ij}}\rVert^{-\eta}}\Bigg]\notag\\
        &\overset{(b)}\approx
        \prod_{q \in Q} 
        \mathbb{E}_{y,x,t} \Bigg[
        \prod_{y_i \in\Phi_{G}\backslash{y_0}} \prod_{x_j \in\tilde{\Phi}_{ED,i,q}}
         \frac{1}{1+s P_q h_{q_{0},q}(T_{ij})\alpha\lVert {y_{i}+x_{ij}}\rVert^{-\eta}} \Bigg]\notag\\
        &\overset{(c)}=\prod_{q \in Q} exp\left(-2\pi \lambda_{G} \int_{0}^{\infty} \big(1-\xi_{q}(s,y)\big) y dy\right) \notag,
      \label{eqkappaint}
   \end{align}
    Following similar steps to $\mathfrak{L}_{I_{intra}}$, (a) is obtained using the independence of channel gains from both the spatial and temporal dimensions and the MGF of the exponential distribution, 
    (b) is an approximation obtained using FKG inequality for $T_{ij}>0$ and extended to $\forall$ $T_{ij}$ through validation by MC simulations, 

    and (c) is obtained using the PGFL of the Matern cluster process \cite{distanceMCP} with
        \begin{align}
            \xi_{q}(s,y)=e^{-\lambda_{ED} a \int_{d_{q-1}}^{d_{q}} \int_{0}^{2\pi} \left(1-{E}_t\left(\frac{1}{1+s P_{q} h_{q_{0},q}(T_{ij})\alpha {\beta(x,y,\theta)}^{-\eta}}\right)\right)xdxd\theta },
        \end{align}
  with $\beta(x,y,\theta)={\sqrt{y^2+x^2-2xycos(\theta)}}$. 
For the case of a highly clustered network we have $x << y$. Using the approximation in Corollary 2 \cite{LoraZijin}, we have $\beta(x,y,\theta)\approx y$ and doing a  Taylor series expansion we obtain:
\begin{align}
   \mathfrak{L}_{I_{inter}} \{s\}   &=\prod_{q \in Q} e^{-2\pi \lambda_{G}\lambda_{ED} a N_{q})\int_{0}^{\infty}  \int_{-\infty}^{+\infty}\left( \frac{s P_{q} h_{q_{0},q}(t)\alpha {y}^{-\eta}}{1+s P_{q}h_{q_{0},q}(t)\alpha {y}^{-\eta}}\right)f_{T_{ij}(t)}ydt dy}\notag\\
   &\overset{(d)}=\prod_{q \in Q}e^{-2\pi \lambda_{G}\lambda_{ED} a N_{q} \frac{\pi(sP_{q}\alpha)^{\frac{2}{\eta}}}{\eta \sin{(\pi\frac{2}{\eta}})}  \int_{-\infty}^{+\infty}(h_{q_{0},q}(t))^{\frac{2}{\eta}}f_{T_{ij}(t)}dt},\notag
    \end{align} 
where (d) is obtained using \cite[(3.241)]{integralsbook}. Recalling Lemma \ref{lemcol}, we evaluate $\int_{-\infty}^{+\infty}(h_{q_{0},q}(t))^{\frac{2}{\eta}}f_{T_{ij}(t)}dt$. For $s=\rho$, we obtain the final expression in (\ref{EqInter}).

\end{document}